\documentclass[letterpaper,11pt,twoside,reqno,nonamelimits]{amsart}
\usepackage[pagebackref,colorlinks=true,urlcolor=blue,linkcolor=blue,citecolor=blue]{hyperref}
\usepackage[top=1.65in,bottom=1.7in,left=1.5in,right=1.5in]{geometry}
\usepackage{times}
\usepackage{amsmath}
\usepackage{amssymb}
\usepackage{amsfonts}
\usepackage{amscd}
\usepackage{amsthm}
\usepackage{amsxtra}

\usepackage[all]{xy}  \CompileMatrices
\usepackage{mathrsfs}
\usepackage{nicefrac}
\usepackage{bbold}
\usepackage{graphicx}
\usepackage{color}
\usepackage{enumerate}
\usepackage{url}

\usepackage{algorithm}

\theoremstyle{plain}
\newtheorem{theorem}{Theorem}
\newtheorem{proposition}[theorem]{Proposition}
\newtheorem*{proposition*}{Proposition}

\newtheorem*{corollary*}{Corollary}
\newtheorem{lemma}[theorem]{Lemma}

\newtheorem*{theorem*}{Theorem}
\newtheorem*{lemma*}{Lemma}
\newtheorem*{conjecture*}{Conjecture}

\newtheorem*{question*}{Question}

\theoremstyle{definition}

\newtheorem*{exercise*}{Exercise}

\theoremstyle{remark}

\newtheorem*{remark*}{Remark}
\newtheorem{remsTh}[theorem]{Remarks}

\newcommand{\subclass}[1]{}

\newcommand{\enumTi}[1]{\renewcommand{\theenumi}{#1}}

\newcommand{\alphenumi}{\enumTi{\alph{enumi}}}

\newcommand{\romenumi}{\enumTi{\roman{enumi}}}

\newcommand{\lt}{\left}
\newcommand{\rt}{\right}

\newcommand{\abs}[1]{{\lt\lvert{#1}\rt\rvert}}

\newcommand{\nfrac}[2]{{\nicefrac{#1}{#2}}}

\newcommand{\NN}{\mathbb{N}}

\newcommand{\RR}{\mathbb{R}}

\DeclareMathOperator{\Prb}{\mathbf{P}}
\DeclareMathOperator{\Exp}{\mathbf{E}}

\DeclareMathOperator{\IndicatorOp}{\mathbf{I}}
\newcommand{\Ind}{\IndicatorOp}

\newlength{\algotabbingwidth}
\newcommand{\xx}{\mathtt{x}}

\begin{document}
\title[Random regular SAT]{On the satisfiability of random regular signed SAT formulas}%
\author{Christian Laus}%
\author{Dirk Oliver Theis}%
\address{DOT \& CL: Fakult\"at f\"ur Mathematik\\
  Otto-von-Guericke-Universit\"at Magdeburg\\
  Universit\"atsplatz~2\\
  39106~Magdeburg\\
  Germany.  \href{http://dirkolivertheis.wordpress.com}{\rm http://dirkolivertheis.wordpress.com}}
\email{theis@ovgu.de}%

\begin{abstract}
  Regular signed SAT is a variant of the well-known satisfiability problem in which the variables can take values in a fixed set $V \subset [0,1]$, and the literals
  have the form ``$\xx \le a$'' or ``$\xx \ge a$'' instead of ``$\xx$'' or ``$\bar\xx$''.
  
  We answer some open question regarding random regular signed $k$-SAT formulas: The probability that a random formula is satisfiable increases with $\abs{V}$; there
  is a constant upper bound on the ratio $m/n$ of clauses~$m$ over variables~$n$, beyond which a random formula is asypmtotically almost never satisfied; for $k=2$
  and $V=[0,1]$, there is a phase transition at $m/n=2$.
\end{abstract}

\date{Tue Dec  6 18:05:36 CET 2011}

\maketitle

\section{Introduction}\label{sec:intro}

Let $V$ be a set with at least two elements, $\mathcal S$ a set of subsets of $V$ called \textit{signs}, and $k$ a positive integer.  For the \textit{signed
  $k$-satisfiability problem,} or \textit{signed $k$-SAT,} one is given as input a set of~$n$ variables $X$ and a formula in \textit{signed conjunctive normal form.}
This means that there is a list of~$m$ \textit{clauses}, each of which is a conjunction ($\land$) of $k$~\textit{(signed) literals} of the form ``$\xx \in
\texttt{S}$'' where $\xx$ is a variable in $X$ and the ``sign'' $\texttt{S}$ is a set in $\mathcal S$.  The question is then whether there exists a satisfying
\textit{interpretation,} i.e., an assignment of values to the variables such that each of the clauses is satisfied.

Historically, signed SAT originated in the area of so-called multi-valued logic, where variables can take a certain number of truth values, not just 0 or~1.  This is
why the set~$V$ is called the \textit{truth-value set.}  We refer the reader to the survey paper~\cite{BeckertHaehnleManya99}, and the references therein.

In the signed $k$-SAT area, the variants where the literals are inequalities have received special attention (see
also~\cite{HochbaumMorenocenteno08,BallersteinTheis13}).  One speaks of \textit{regular} signed $k$-SAT or just \textit{regular $k$-SAT} (\textit{$k$-rSAT} for
short), if~$V$ is a (linearly) ordered set, and the allowed signs are $\{ x \mid x \le a\}$ and $\{ x \mid x \ge a\}$, $a\in V$.  The literals then have the form
``$\xx \le a$'' or ``$\xx \ge a$''.  Clearly, the satisfiability of the formula depends only on~$v := \abs{V} \in \{2,3,\dots,\infty\}$ rather than on the set~$V$
itself, so we always assume that $V \subset\RR$ and $\min V = 0$, and $\max V = 1$.

This setting includes as a special case the classical satisfiability (SAT) problem: choose for~$V$ the 2-element set $\{0,1\}$, and use the signed literals $\xx \ge
1$ and $\xx \le 0$ to represent the classical SAT literals $\xx$ and $\bar x$, respectively.

This paper is about $k$-rSAT formulas drawn at random from all such formulas with a fixed truth-value set~$V$.  We allow that either $v < \infty$, or, in the limit,
$V = [0,1]$.
These random formulas are studied for $m=cn$, for a fixed constant~$c$, in the limit $n\to \infty$.
Literals $\xx \le 1$ and $\xx \ge 0$ are \textit{innocuous:} they only affect the number of clauses which remain to be satisfied.
We will drawn formulas uniformly at random from the set of all signed $k$-SAT formulas with $n$~literals, $m$~clauses and truth-value set~$V$, which do not contain
any innocuous literal.\footnote{%
  If innocuous literals are allowed in a random formula, the number of clauses which contain an innocuous literal is distributed like a binomial variable with
  parameters $m=\Theta(n)$ and $1/v=O(1)$, and as such is with high probability $O(\sqrt n)$, so that the ratio $c=\nfrac{m}{n}$ is unaffected for $n\to\infty$.
  Hence, forbidding innocuous literals does not change asymptotic results.
}%

Based on computational investigations of the satisfiability of uniformly generated random 3-rSAT instances, Many\`a et al.~\cite{ManyaBejarEscaladaimaz98} have made
a number of observations and conjectures.  Most importantly, they observed a phase transition phenomenon similar to the one in classical SAT (see,
e.g.,~\cite{Friedgut99,AchlioptasPeres2004} and the references therein).  They interpret their results as supporting the existence of a threshold $c = c_k(v)$, for
$k=3$, with the following properties:
\begin{itemize}
\item[(i)] the most computationally difficult instances tend to be found when the ratio $\nfrac{m}{n}$ is close to~$c_k(v)$;
\item[(ii)] there is a sharp transition from satisfiable to unsatisfiable instances when the ratio $\nfrac{m}{n}$ crosses the threshold;
\item[(iii)] $c_k(v)$ is nondecreasing in the number of truth-values~$v$.
\end{itemize}
Their results are confirmed and extended by other papers exploring uniformly random 3-rSAT
instances~\cite{BejarManya99phase,BeckertHaehnleManya99,BejarManyaCabiscolFernandezGomes07}.
From their computational data, B\'ejar et al.~\cite{BejarManya99phase,BejarManyaCabiscolFernandezGomes07} surmise that
\begin{enumerate}[(i)]
\item[(iv)] the threshold $c_k(v)$ increases logarithmically with~$v$,
\end{enumerate}
and prove that for
\begin{equation}\label{eq:BejarManya-bound}
  c > \log_{\nfrac87}(v)
\end{equation}
a random 3-rSAT formula with $\nfrac{m}{n} = c$ is asymptotically almost never (a.a.n., as $n\to\infty$) satisfiable.
Their proof (and bound) resembles that for classical $k$-SAT ($2^k \log 2$).

\subsection*{Our contributions.}
In this paper, we prove~(ii) for $k=2$ and $V=[0,1]$; establish~(iii) for all~$k$; improve the bound~\eqref{eq:BejarManya-bound} for large~$v$; and falsify~(iv).

To elaborate, for~(iii), we show that the probability that a random formula is satisfied increases with~$v$.  In particular, the probability that a random formula on
a finite truth-value set is satisfiable is bounded from above by the probability that a random formula with $V=[0,1]$ is satisfiable.  Thus, if $c_3(v)$ increased
logarithmically with~$v$, then for any finite~$c$, a random formula with truth-value set~$[0,1]$, $n$ variables and $m=cn$ clauses would be asymptotically almost
surely (a.a.s.)\ satisfiable.

We then prove the following.

\begin{theorem}\label{thm:general-ub}
  If~$c > 1$ is such that 
  \begin{equation*}
    kc\lt( 1 - 2^{-k} \rt)^{c-1} < 1,
  \end{equation*}
  then a random $k$-rSAT formula with~$n$ variables, $m=cn$ clauses, and $V = [0,1]$ is a.a.n.\ satisfiable.
\end{theorem}

This improves on B\'ejar et al.'s~\cite{BejarManyaCabiscolFernandezGomes07} bound~\eqref{eq:BejarManya-bound} for large values of~$v$.  Most notably, it gives a
finite upper bound for all~$v$ and thus disproves~(iv).  In particular, Theorem~\ref{thm:general-ub} implies the following.

\begin{corollary*}%
  For all~$V$, a random $3$-rSAT formula with~$n$ variables and~$m=cn$ clauses is a.a.n.\ satisfiable, if $c \ge 36.1$.
  \qed
\end{corollary*}

We then move on to study 2-rSAT.  Here, Theorem~\ref{thm:general-ub} gives an upper bound of apx.~12.664 beyond which a random 2-rSAT is a.a.n.\ satisfiable.  To
prove a lower bound beneath which such a formula is satisfiable, we use a result by Chepoi et al.~\cite{ChepoiCreignouHermannSalzer10}, who prove a characterization
of non-satisfiable signed 2-SAT instances based on a digraph certificate, in the spirit of Aspvall, Plass, and Tarjan's famous result for classical
2-SAT~\cite{AspvallPlassTarjan79}.  Using Chepoi et al.'s characterization we prove the following.

\begin{theorem}\label{thm:2-lb}\label{thm:2-ub}
  A random $2$-rSAT formula with~$n$ variables, $m=cn$ clauses, and $\abs{V} = \infty$ is
  \begin{enumerate}[(a)]
  \item a.a.s.\ satisfiable, if $c < 2$, and
  \item a.a.n.\ satisfiable, if $c > 2$.
  \end{enumerate}
\end{theorem}

The improved upper bound here comes from a concentration result.  This theorem shows that, as for classical random $k$-SAT, $k$-rSAT exhibits a threshold behaviour
if $k=2$.

The main difficulty in the last theorem for 2-rSAT, as compared to classical 2-SAT, comes from the fact that there is an infinite number of possible literals --- as
opposed to the total of $2n$ possible literals for classical SAT.  In our proof, we make use of the following trick.  When conditioning on the number of times~$R_j$
each variable~$\xx_j$ occurs, the structure one needs to analyze has some resemblance to the configuration model for random (multi) graphs with prescribed
degrees~$R_j$.  This way, Chv\'atal and Reed's~\cite{ChvatalReed92} approach for classical 2-SAT can be adapted.

\subsection*{Organization of the paper.}

The remainder of the paper is organized as follows.  In the next section, we discuss the random model, and variants of it, in the necessary details and prove the
monotonicity of the probability of satisfiability mentioned above.  Section~\ref{sec:pf-general-ub} contains the proof of Theorem~\ref{thm:general-ub}.  Sections
\ref{sec:pf-2-lb} and~\ref{sec:pf-2-ub} hold the proof of Theorem~\ref{thm:2-lb}.  In the final section, we discuss a few open questions.

\section{Basics about random $k$-{\textrm{\rm r}}SAT}\label{sec:basics}

In this section, we discuss variants of the random model which we need.  Then we will prove some basic facts about $k$-rSAT.

We will think of a random formula as being constructed as follows.  First of all, we assume that the truth-value set $V$ is a subset of the unit interval $[0,1]$
which is symmetric (i.e., $V = 1-V$) and which contains both $0$ and~$1$.

Now we take an ``empty'' formula, i.e., we have $m\times k$ empty slots.  Each slot will be filled by a triple $(\xx,\varrho,a)$ where $\xx$ is one of the variables,
$\varrho$ is a comparison relation ``$\le$'' or ``$\ge$'', and $a$ is in $V\setminus\{1\}$.  The interpretation of such a triple is that, if
$\varrho=\mbox{``$\le$''}$, then we have the condition $\xx \le a$, whereas, if $\varrho=\mbox{``$\ge$''}$, we have the condition $\xx \ge 1-a$.  By this
construction, we exclude the cases of the inequalities $\xx \ge 0$ and $\xx \le 1$, which are meaningless because they do not constrain~$\xx$.  The part
$(\varrho,a)$ is referred to as the constraint part of the literal.

For each slot, the three parts of the triple are chosen independently from each other.  The selection of the right hand sides and the comparison relations is done
independently for all slots.
For the variables, there are several possibilities.
First of all, their selection may either be chosen independently for all slots (i.e., allowing a clause to contain more than one slot with the same variable), or
indepdently for all slots but conditioning on the~$k$ variables occuring in a clause being distinct.  In the second case, the event on which we condition is
asymptotically bounded away from~$0$.  Hence, as far as a.a.s.\ statements about random formulas are concerned, the two possibilities for the random selection of the
variables are equivalent.
We denote by $F_k(n,m,v)$ a random $k$-rSAT formula with truth-value set of cardinality~$v$ in which, for each clause, the variables in the slots are distinct; by
$F'_k(n,m,v)$, we denote a random formula where the variables can occur multiple times in the same clause.

Secondly, we may choose the variables conditioning on the number of times each variable occurs in the formula.  For a random formula~$F$, let the random
variable~$R_j$ denote the number of slots containing variable~$\xx_j$.  Clearly, we have
\begin{equation}\label{eq:sum-var-occ}
  \sum_{j=1}^n R_j = km.
\end{equation}
If, when choosing the variables, we allow a clause to contain more than one slot with the same variable, then the $R := (R_j)_{j=1,\dots,n}$ has
multinomial distribution, i.e., for all $r\in \NN^n$ with $\sum_j r_j = km$ we have
\begin{equation}\label{eq:r-distrib}
  \Prb[ R = r ] = \frac{\binom{km}{r_1,\dots,r_n}}{n^{km}}.
\end{equation}
This is same as the $R_j$ being independent Poison with mean $km/n$ conditioning on~\eqref{eq:sum-var-occ}.

When constructing a random formula, we may reverse this view: We may condition on the values of~$R$.  This amounts to pretending that, for $j=1,\dots,n$, there are
$R_j$ distinguishable copies of variable~$\xx_j$, and the $km$ variable copies are assigned to the $km$ slots randomly.


\subsection*{Monotonicity}
We now come the some basic facts about random $k$-rSAT formulas.  We start the monotonicity property of rSAT formulas mentioned in the introduction.
Denote by
\begin{equation}\label{eq:def-p}
  p_k(n,m,v) := \Prb[ F_k(m,n,v) \text{ is satisfiable}]
\end{equation}
the probability that a random $k$-rSAT formula with $m$ clauses on $n$ variables and truth-value set of cardinality $\abs{V}$ is satisfiable.  We will habitually omit
the~$k$.  Naively speaking, increasing $\abs{V}$ increases the possible choices for the variables, so we would guess that $p(n,m,v)$ increases with~$v$.  That is in
fact the case.

The easiest setting in which we can visualize this phenomenon is, if we suppose that the right hand sides are of the form
\begin{equation*}
  A = \sum_{i=1}^\lambda B_i 2^{-i},
\end{equation*}
where the $B_i$, $i=1,2,3,\dots$ are independent Bernoulli random variables with $\Prb[B_i = 1] = \nfrac12$, and $\lambda$ is either finite --- in which case
$\abs{V} = 2^{\lambda}+1$ --- or $\lambda=\infty$, in which case $V=[0,1]$.  Now note that increasing~$\lambda$ increases~$A$.  But this weakens the inequalities
constraining the variables, and thus makes the formula ``more satisfiable''.
An only slightly more technical argument proves this monotonicity fact for general $\abs{V}$.

\begin{lemma}\label{lem:p_monotone}
  For every $k,n,m$, the following hold.
  \begin{enumerate}[(a)]
  \item For every $v$, we have $\displaystyle p_k(n,m,v) \le p_k(n,m,v+1)$.
  \item For every $v$, we have $\displaystyle p_k(n,m,v) \le p_k(n,m,\infty)$.
  \item We have $\displaystyle \lim_{v\to\infty} p_k(n,m,v) = p_k(n,m,\infty)$.
  \end{enumerate}
\end{lemma}
\begin{proof}
  \textit{(a). }
  Suppose we have $V = \{u/(v-1) \mid u=0,\dots,v-1\}$.  For a random constraint, a value in $V\setminus\{1\}$ is drawn uar.  We would like increase $v\leadsto v+1$.
  Firstly, we scale all values by the factor $(v-1)/v$.  This does not influence satisfiability.  Secondly, we increase $u/v$ to $(u+1)/v$ with probability
  $(u+1)/v$, $u=0,\dots,v-2$.  This yields the uniform distribution on $V'\setminus\{1\} = \{ u/v \mid u=0,\dots,v-1\}$.

  Note that increasing a in a literal $(\xx,\rho,a)$ will never make a satisfiable formula unsatisfiable: indeed, the set of satisfying interpretations stays the
  same.  However, if a formula is not satisfiable, whenever the value for~$a$ in a literal is increased, there is a possibility that the formula becomes satisfiable.
  Thus, the probability that a random formula is satisfiable increases with $v\leadsto v+1$.
  
  \textit{(b). }
  To prove~(b), by~(a), it suffices to consider the following sets, for $\lambda = 0,1,2,\dots,\infty$:
  \begin{equation*}
    V_{\lambda} := \Bigl\{  \sum_{i=1}^\lambda  B_i 2^{-i} \mid B \in \{0,1\}^\NN \Bigr\} \cup \{1\},
  \end{equation*}
  i.e., $V_0 = \{0,1\}$, $V_1 = \{0,\nfrac12,1\}$, $V_2 = \{0,\nfrac14,\nfrac12,\nfrac34,1\}$, \dots, $V_\infty = [0,1]$.
  
  To prove~(b), we now use the method of deferred decisions.  We draw~$B=(B_1,B_2,\dots)$ randomly regardless of the value of~$\lambda$.  For a random formula
  $F_k(n,m,2^\lambda+1)$, the $B_1,\dots,B_\lambda$ have been exposed.  Increasing $\lambda\leadsto \infty$ amounts to exposing all remaining $B_{i}$,
  $i=\lambda+1,\lambda+2,\dots$, and adjusting the literals of the formula accordingly.  But this can only increase the sum, and thus, modifying the literals of a
  formula in this way can only turn a not satisfiable formula into a satisfiable one, and thus can only increase $p_k(n,m,\cdot)$.

  \textit{(c). }
  We use the same approach as in~(b).  Suppose that a formula $F := F_k(n,m,\infty)$ is satisfiable.  We then find a finite $\lambda$ such that truncating the sums
  at the $2^{-\lambda}$-term already yields a satisfiable formula.  First of all, we may assume that the literals $(\xx,\rho,a)$ all have distinct values~$a$.  Let
  $\lambda_-$ be the largest number such that there are literals $(\xx,\rho,a)$ and $(\xx',\rho',a')$ for which the sums in $a$ and~$a'$ coincide up to the
  $2^{-\lambda_-}$-term.  Then $\lambda_-$ is finite.  Letting $\lambda := \lambda_- +1$, we see that in truncated random formula, the constraint parts of the
  literals have the same relative ordering as in the original formula.  Hence, the truncated formula is satisfiable.
  
  Thus, every satisfiable formula for $\lambda=\infty$ becomes satisfiable already at a finite value for~$\lambda$.  This proves the stated convergence.
\end{proof}

Since the existence of a threshold is, as of now, conjectural, we let, for $v\in \{2,3,\dots,\infty\}$
\begin{align*}
  c^-_k(v) &:= \sup\{ c \mid F_k(n,cn,v) \text{ a.a.s.\ satisfiable} \}\text{, and}\\
  c^+_k(v) &:= \inf\{ c \mid F_k(n,cn,v) \text{ a.a.n.\ satisfiable} \}.\\
\end{align*}
The existence of a threshold is then equivalent to $c^-_k(v) = c^+_k(v)$; cf.~Fig.~\ref{fig:transition}.

\begin{figure}[htp]
  \centering%
  \scalebox{.6666666}{\input{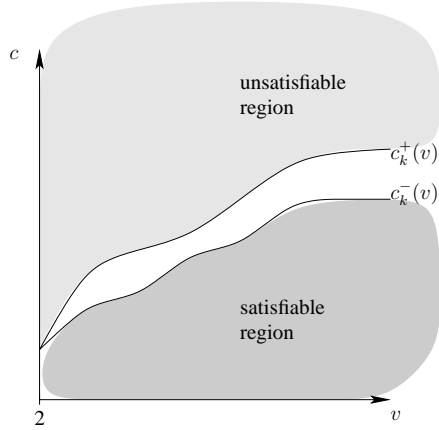}}%
  \caption{Transition from satisfiable to unsatisfiable for increasing~$c$.  Note that both $c^+_k$ and $c^-_k$ converge to a finite value as $v\to\infty$ by
    Proposition~\ref{prop:ckv_monotone} and Theorem~\ref{thm:general-ub}.}\label{fig:transition}%
\end{figure}

From Lemma~\ref{lem:p_monotone}, we immediately derive the following.

\begin{proposition}\label{prop:ckv_monotone}
  For every~$k$ we have the following.
  \begin{enumerate}[(a)]
  \item $c^-_k(v)$ and $c^+_k(v)$ are both nondecreasing with the cardinality~$v$ of the truth-value set.
  \item $c^-_k(v) \le c^-_k(\infty)$ and $c^+_k(v) \le c^+_k(\infty)$.
  \item $\lim_{v\to\infty} c^-_k(v) = c^-_k(\infty)$ and $\lim_{v\to\infty} c^+_k(v) = c^+_k(\infty)$.
  \end{enumerate}
\end{proposition}
\begin{proof}
  (a) and (b) follow immediately from their counterparts in Lemma~\ref{lem:p_monotone}.  As for~(c), let $c^+ := \lim_{v\to\infty} c^+_k(v)$ and assume that $c^+ <
  c^+_k(\infty)$.  Then, for $c$ with $c^+ < c < c^+_k(\infty)$, we have that $p(n,cn,\infty)$ does not converge to~$0$, so there is a sequence
  $(n_\ell)_{\ell=1,2,\dots}$ for which $p(n_\ell,cn_\ell,\infty)$ is bounded away from~0.  But $\lim_n p(n,cn,v) = 0$, contradicting part~(c) of
  Lemma~\ref{lem:p_monotone}.

  Similarly, let $c^- := \lim_{v\to\infty} c^-_k(v)$ and assume that $c^- < c^-_k(\infty)$.  Then, for $c$ with $c^- < c < c^-_k(\infty)$, we have $\lim_n
  p(n,cn,\infty) = 1$.  But for all $v$, we have that $p(n,cn,v)$ is bounded away from~1.  We obtain a contradiction in the same way as above.
\end{proof}

\subsection*{Ancilliary}
We conclude the section with the following trivial lemma.

\begin{lemma}\label{lem:prb-sat-dj}
  If~$V=[0,1]$, then the following holds.
  \begin{enumerate}[(a)]
  \item For every fixed $x\in[0,1]$, the probability that a random constraint defined by $(\varrho,A)$ is satisfied, is $\nfrac12$.
  \item For every to constraints $(\varrho,A)$, $(\varrho',A')$, the probability that there is no $x\in V$ satisfying both (i.e., the signs are disjoint) is
    $\nfrac14$.
  \end{enumerate}
  \qed
\end{lemma}

\section{Proof of Theorem~\ref{thm:general-ub}}\label{sec:pf-general-ub}

In this section, we prove Theorem~\ref{thm:general-ub}.  By Lemma~\ref{lem:p_monotone}, it suffices to prove the statement for the case when the right hand sides of
the formulas are drawn uar from $[0,1]$.  In particular, we can assume that no two literals have the same right hand side.

An interpretation $\xx \to x$ of a formula~$F$ is called \textit{tight,}, if for every variable $\xx_j$, there one of the two literals ``$\xx_j \le x_j$'' or
``$\xx_j \ge x_j$'' occurs in~$F$.  In other words, only the right hand sides of the inequalities are allowed as values for the variables.  The following fact is
trivial.

\begin{lemma}
  There exists a satisfying interpretation if, and only if, there exists a satisfying tight interpretation.
  \qed
\end{lemma}

For a random formula~$F$, denote by $Y = Y_F$ the number of satisfying tight interpretations.  To prove Theorem~\ref{thm:general-ub}, we compute the expectation
of~$Y$.

When sampling a random formula, we condition on $R=(R_j)_{j=1,\dots,n}$ as discussed above.  For $j=1,\dots,n$ and $\ell=1,\dots,R_j$, let $A_{j,\ell}$ be the
(random) right hand side in the slot containing the $\ell$th copy of the variable $\xx_j$.  For every $\ell \in \prod_{j=1}^n \{1,\dots,R_j\}$, we construct an
interpretation $\xx\to x(\ell)$ by letting $x_j(\ell) := A_{j,\ell(j)}$.  With the sum below extending over all $\ell \in \prod_{j=1}^n \{1,\dots,R_j\}$, we have
\begin{equation*}
  Y
  = 
  \sum_{\ell} \Ind[ x(\ell) \text{ satisfies } F ].
\end{equation*}

For every fixed $\ell$, we can estimate the probability of the event that $x(\ell)$ satisfies~$F$.  Indeed, $n$~literals will be ``automatically'' satisfied, namely
for every variable~$\xx_j$ the one containing the $\ell(j)$th copy of~$\xx_j$.  Since the right hand sides are drawn independently, for each of the remaining
literals, the probability of being satisfied by~$x$ is $\nfrac12$, by Lemma~\ref{lem:prb-sat-dj}.  The automatically satisfied literals cover at most~$n$ clauses,
which leaves $(c-1)n$ clauses, each of which contains exactly~$k$ of the remaining literals.  Conditioned on the assignment~$X$ of the variable copies to the slots,
the probability that all of these clauses are satisfied is thus at most
\begin{equation*}
  \bigl( 1-2^{-k} \bigr)^{(c-1)n}.
\end{equation*}
Moreover, the event that all the remaining clauses are satisfied depends only on the constraint part of the literals and is thus independent from~$R$
  
Now we can compute the expected number of satisfying tight interpretations.
\begin{multline*}
  \Exp Y
  =
  \Exp \biggl( \Exp\Bigl( \sum_{\ell} \Ind[ x(\ell) \text{ satisfies } F ] \Bigm| R \Bigr) \biggl)
  \\
  =
  \Exp \biggl( \Exp\Bigl( \sum_{\ell} \Exp( \Ind[ x(\ell) \text{ satisfies } F ] \mid X ) \Bigm| R \Bigr) \biggl)
  \\
  \le
  \bigl( 1-2^{-k} \bigr)^{(c-1)n} \Exp \biggl( \Exp\Bigl( \sum_{\ell} 1 \Bigm| R \Bigr) \biggl)
  \\
  =
  \bigl( 1-2^{-k} \bigr)^{(c-1)n} \Exp \Bigl( \prod_{j=1}^n R_j \Bigl)
  \le
  \bigl( 1-2^{-k} \bigr)^{(c-1)n} (\nfrac{km}{n})^n
  =
  \Bigl(  kc\bigl( 1-2^{-k} \bigr)^{c-1}  \Bigr)^n.
\end{multline*}
From this, Theorem~\ref{thm:general-ub} follows by Markov's inequality.

\section{Proof of Theorem~\ref{thm:2-lb}({\textrm{\rm a})}}\label{sec:pf-2-lb}

In this section, we prove Theorem~\ref{thm:2-lb}.  For this, we use the Aspvall-Plass-Tarjan-style~\cite{AspvallPlassTarjan79} characterization of non-satisfiable
signed 2-SAT formulas by Chepoi et al.~\cite{ChepoiCreignouHermannSalzer10} together with Chv\'atal and Reed's~\cite{ChvatalReed92} trick of counting ``bicycles''.

An \textit{$\ell$-bicycle} contained in a 2-rSAT formula~$F$ is a sequence of $2\ell$ literals
\begin{equation*}
  w^f_0,w^t_1,w^f_1\dots,w^t_\ell,w^f_\ell,w^t_{\ell+1}
\end{equation*}
together with two (not necessarily distinct) numbers $i_0\in\{2,\dots,\ell\}$ and $i_1\in\{1,\dots,\ell-1\}$ such that
\begin{enumerate}[(bc1)]
\item\label{enum:bicycle:vars-dtct} the variables in $w^t_1,w^t_2,\dots,w^t_\ell$ are all distinct;
\item\label{enum:bicycle:vars-traverse} the variables in the two literals $w^t_i,w^f_i$ are the same, for $i=1,\dots,\ell$;
\item\label{enum:bicycle:vars-ends} the variable of $w^f_0$ is the same as the one of $w^t_{i_0}$, and the variable of $w^f_{\ell+1}$ is the same as the one of
  $w^t_{j_1}$;
\item\label{enum:bicycle:clauses} for each $i=0,\dots,\ell$, ``$w^f_{i} \lor w^t_{i+1}$'' is a clause in~$F$;
\item\label{enum:bicycle:traverse} for each $i=1,\dots,\ell$, the constraint parts of $w^t_i,w^f_i$ are disjoint.
\end{enumerate}

The following is an immediate adaption of Chv\'atal and Reed's proof in~\cite{ChvatalReed92} to Chepoi et al.'s~\cite{ChepoiCreignouHermannSalzer10} variant, for the
signed case, of Aspvall et al.'s~\cite{AspvallPlassTarjan79} characterization of non-satisfiable 2-SAT formulas.

\begin{lemma}
  Every unsatisfiable 2-rSAT formula contains an $\ell$-bicycle, for some $\ell \ge 2$.
  \qed
\end{lemma}

As in the previous section, we will make use of the fact that, for $\abs{V}=\infty$, we may assume that no two
literals have the same right hand side.

As above, let $R=(R_j)_{j=1,\dots,n}$ count the occurences of the variables in the random formula~$F=F'_2(n,m,\infty)$.  Conditioned on~$R$, we can recover the
distribution of~$F$ as follows.
Let there be~$n$ \textit{buckets} $B_1,\dots,B_n$; bucket~$B_j$ contains~$R_j$ \textit{points.}  Note that there is an even number $2m$ of points.  A perfect
matching of the points corresponds to selecting two variables for each clause of the formula.  Hence, drawing a matching at random and, independently, drawing a
random constraint part for each point, gives us a random formula.  The distribution is the same as that for random formulas, conditioned on~$R$.

We now count the number of $\ell$-bicycles in~$F$ by counting the corresponding matchings.  We start with the following easy lemma.

\begin{lemma}\label{lem:poisson-trick}
  Let $d_1,\dots,d_n$ be nonnegative integers.  Then
  \begin{enumerate}
  \item $\displaystyle \Exp \prod_{j=1}^n (R_j)_{d_j} \le (2c)^{\sum_j d_j} $
  \item If all $\sum_j d_j = O(1)$, then $\Exp \prod_{j=1}^n (R_j)_{d_j} = (1+o(1)) (2c)^{\sum_j d_j}$.
  \end{enumerate}
\end{lemma}
\begin{proof}
  This is a direct computation using~\eqref{eq:r-distrib}.
\end{proof}

We are now ready to complete the proof.

\begin{proof}[Proof of Theorem~\ref{thm:2-lb}]
  Let $b=(b_1,\dots,b_\ell) \in \{1,\dots,n\}^\ell$ be the choice of the $\ell$ distinct variables in condition~(\ref{enum:bicycle:vars-dtct}) of the definition of a
  bicycle above.  For each one of these~$b$, we have to choose $i_0 \in \{2,\dots,\ell\}$ and $i_1 \in \{1,\dots,\ell-1\}$.


  Letting $d_{b_i} := 2$, $i=1,\dots,\ell$, and $d_j := 0$ for each~$j$ not occuring in $(b_1,\dots,b_\ell)$, the number of choices for the matching edges between
  the buckets corresponding to the clauses ``$w^f_{i} \lor w^t_{i+1}$'', $i=1,\dots,\ell-1$, conditioned on~$R$, is $\prod_{j=1}^n (R_j)_{d_j}$.  For the clauses
  ``$w^f_{0} \lor w^t_{1}$'' and ``$w^f_{\ell} \lor w^t_{\ell+1}$'', the number of choices depend on whether $i_0=\ell$, or $i_1=0$, or both, or neither.  For each
  choice of $i_0$ and $i_1$, if we change the $d_j$ to count the number of times the variable~$\xx_j$ occurs in the bicycle, the number of choices is $\prod_{j=1}^n
  (R_j)_{d_j}$.  There are at most $\ell^2$ choices for the $i_0$ and $i_1$, and we have $\sum_j d_j = 2(\ell+1)$, so that, by Lemma
  by Lemma~\ref{lem:poisson-trick}, the expectation of the number of choices for the matching edges in the bicycle for fixed $b$ and $i_0, i_1$ is at most
  $(2c)^{2(\ell+1)}$, whereas the total number of choices for these matching edges equals $(2m-1)(2m-3)\dots(2m-2\ell-1)$.
  
  There are $(n)_\ell$ possible choices of~$b$, and
  the probability of disjointness in~(\ref{enum:bicycle:traverse}) is $\nfrac14$ for each variable (by Lemma~\ref{lem:prb-sat-dj}), or $4^{-\ell}$ for the whole
  bicycle.
  
  Thus, denoting by $Y_\ell$ the number of $\ell$-bicycles in a random formula and by $X_{(b,i_0,i_1)}$ the indicator variable of the event that a bicycle with these
  parameters exists, we may compute as follows:
  \begin{multline*}
    \Exp Y_\ell
    = \Exp\lt( \Exp\Bigl( \sum_{b} \sum_{i_0,i_1} X_{(b,i_0,i_1)} \Bigm| R \Bigr) \rt)
    = \sum_{b} \Exp\Bigl( \sum_{i_0,i_1} \Exp( X_{(b,i_0,i_1)} \mid R ) \Bigr)
    \\
    \le \sum_{b}  4^{-\ell} \ell^2 (2c)^{2(\ell+1)} \frac{ 1 }{{(2m-1)(2m-3)\dots(2m-2\ell-1)}} \biggr)
    \\
    = 4^{-\ell} \ell^2 (2c)^{2(\ell+1)} \frac{ (n)_\ell }{{(2m-1)(2m-3)\dots(2m-2\ell-1)}}.
  \end{multline*}
  For the fraction, we use ad-hoc estimates.  Noting that $\ell \le n$, we have $2m-2\ell \ge 2(c-1)n$.  By the monotonicity property, Lemma~\ref{lem:p_monotone}, it
  suffices to prove the theorem for $c>1$, in which case $2m-2\ell-1 =: \omega(n) \to \infty$.
  From Stirling's formula, we see that 
  \begin{multline*}
    \frac{ (n)_\ell }{{(2m-1)(2m-3)\cdots(2m-2\ell-1)}}
    =
    \frac{1}{\omega\, (2c)^\ell} \, \prod_{j=0}^{\ell-1} \frac{n-j}{n- \nfrac{(j+\nfrac12)}{c}}
    \\
    \le 
    \frac{1}{\omega\, (2c)^\ell} \, \prod_{j=0}^{\frac{1}{(c-1)2}} \frac{n-j}{n- \nfrac{(j+\nfrac12)}{c}}
    = (1+o(1)) \frac{1}{\omega\, (2c)^\ell}
  \end{multline*}
  
  Summing over~$\ell$, we see that
  \begin{equation*}
    \sum_{\ell=2}^n Y_\ell
    \le \frac{2c}{\omega} \sum_{\ell} \ell^2(\nfrac{c}{2})^\ell = O_c\Bigl( \frac{1}{\omega} \Bigr),
  \end{equation*}
  for $c < 2$ (the constant in the big-O depends on~$c$).  Thus, the expected number of bicycles of arbitrary length is $O_c(\nfrac1n)$.  From this, the statement of
  the theorem follows.
\end{proof}

\section{Proof of Theorem~\ref{thm:2-ub}({\textrm{\rm b})}}\label{sec:pf-2-ub}

As in the previous section, to prove Theorem~\ref{thm:2-ub}(b), we adapt the approach of Chv\'atal and Reed~\cite{ChvatalReed92}: Prove the non-satisfiability by
establishing the existence of an obstruction by the second moment method.

For an even integer $\ell \ge 6$, we an \textit{$\ell$-snake} consists of a selection of $\ell$ distinct variables $\xx_{b_1},\dots,\xx_{b_\ell}$, and clauses
``$(\xx_{b_i},\rho_i,a_i) \lor (\xx_{b_{i+1}},\rho'_{i+1},a'_{i+1})$'', $i=0,\dots,\ell$, with $b_0 := b_{\ell/2} =: b_{\ell+1}$ such that
\begin{enumerate}[(sk1)]
\item the constraint parts $(\rho'_{i},a'_{i})$ and $(\rho_{i},a_{i})$ are disjoint, for $i=1,\dots,\ell$.
\item the constraint parts $(\rho'_{\ell+1},a'_{\ell+1})$ and $(\rho_{0},a_{0})$ are disjoint.
\item the constraint parts $(\rho'_{\ell/2},a'_{\ell/2})$ and $(\rho'_{\ell+1},a'_{\ell+1})$ are disjoint; as well as the constraint parts
  $(\rho_{\ell/2},a_{\ell/2})$ and $(\rho_{0},a_{0})$ are disjoint.
\end{enumerate}

\begin{lemma}
  If there exists an $\ell$-snake, then the formula is not satisiable.
\end{lemma}
A snake gives rise to a srongly connected component in the digraph of the formula defined by Chepoi et al.~\cite{ChepoiCreignouHermannSalzer10}, which contains a
literal as well as its negation.  Here, we give the elementary proof of the lemma.
\begin{proof}
  Assume that a satisfying interpretation $\xx \to x$ exists.  We prove that the literal $(\xx_{b_{\ell/2}},\rho'_{\ell/2},a'_{\ell/2})$ can be neither satisfied nor
  violated by~$x$.  Suppose it were satisfied.  Then, by disjointness of the constraint parts, $(\xx_{b_{\ell/2}},\rho_{\ell/2},a_{\ell/2})$ must be violated by~$x$.
  Since there is a clause ``$(\xx_{b_{\ell/2}},\rho_{\ell/2},a_{\ell/2}) \lor (\xx_{b_{\ell/2+1}},\rho'_{\ell/2+1},a'_{\ell/2+1})$'', the later literal must be
  satisfied by~$x$.  Proceeding in this fashion, it follows that $(\xx_{b_{\ell+1}},\rho'_{\ell+1},a'_{\ell+1})$ is satisfied, but, since $b_{\ell+1} = b_{\ell/2}$
  and by disjointness, this implies that $(\xx_{b_{\ell/2}},\rho'_{\ell/2},a'_{\ell/2})$ is violated, a contradiction.
    
  Suppose that $(\xx_{b_{\ell/2}},\rho'_{\ell/2},a'_{\ell/2})$ is violated by~$x$.  Since there is a clause ``$(\xx_{b_{\ell/2-1}},\rho'_{\ell/2+1},a'_{\ell/2-1})
  \lor (\xx_{b_{\ell/2}},\rho_{\ell/2},a_{\ell/2})$'', the literal $(\xx_{b_{\ell/2-1}},\rho'_{\ell/2+1},a'_{\ell/2-1})$ must be satisfied by~$x$.  Proceeding as
  above, we conclude that $(\xx_{b_0},\rho_{0},a_{0})$ is satisfied by~$x$, and hence, since $b_0=b_{\ell+1}$, by disjointness,
  $(\xx_{b_{\ell+1}},\rho'_{{\ell+1}},a'_{{\ell+1}})$ is violated.  Since ``$(\xx_{b_{\ell}},\rho_{\ell+1},a_{\ell}) \lor
  (\xx_{b_{\ell+1}},\rho'_{\ell+1},a'_{\ell+1})$'' is a clause, $(\xx_{b_{\ell}},\rho_{\ell},a_{\ell})$ is satisfied, and hence, eventually, so is
  $(\xx_{b_{\ell/2}},\rho_{\ell/2},a_{\ell/2})$.  By disjointness, then, $(\xx_{b_0},\rho_{0},a_{0})$ is violated, a contradiction.
\end{proof}

Fix a $c > 2$, let
\begin{equation*}
  \ell := 2   \lt\lceil   \frac{\log n}{\log(\nfrac{c}{2})}   \rt\rceil,
\end{equation*}
and denote by $X$ the number of $\ell$-snakes.  We will compute the expectation of~$X$ and of its second factorial moment $X(X-1)$, and find that $\Exp X =
\Omega(n)$ and $\Exp X(X-1) = O( (\Exp X)^2 )$.  From this, by Chebyshev's inequality, we conclude that the probability that no $\ell$-snake exists is $O(1/\Exp X) =
O(\nfrac{1}{n})$.  Thus, a random formula $F_2(n,cn,\infty)$ is a.a.n.\ satisfiable.

The following two lemmas comtain the computations of the moments.

\begin{lemma}
  $\displaystyle \Exp X = \Omega(n)$.
\end{lemma}
\begin{proof}
  As in the previous section, we have
  \begin{equation*}
    \Exp(X\mid R)
    = \sum_{b} \frac{  (R_{b_{\ell/2}}-2)(R_{b_{\ell/2}}-3) \prod_{i=1}^\ell (R_{b_i})_2  }{  4^{\ell+3} (2m-1)\cdots(2m-2\ell-1)  },
  \end{equation*}
  where the sum extends over all possible choices $b\in\{1,\dots,n\}^{\ell}$ identifying the $\ell$ distinct variables in a snake.  Thus, by
  Lemma~\ref{lem:poisson-trick}, we have
  \begin{equation*}
    \Exp X
    = \Theta\Bigl( (\nfrac{c}{2})^\ell \tfrac{1}{2m-2\ell-1} \Bigr)
    = \Theta( \nfrac{1}{n} ).
  \end{equation*}
\end{proof}

\begin{lemma}
  $\displaystyle \Exp X(X-1) = O((\Exp X)^2)$.
\end{lemma}
\begin{proof}
  We have to compute the expectation of the sum
  \begin{equation*}
    \sum_S \sum_{S'\ne S} X_{S} X_{S'}
  \end{equation*}
  where the sums range over all possible snakes~$S$ and $S'$, respectively, and $X_S$ denotes the indicator variable of the event that the snake~$S$ exists.  Taking
  the expectation, it can be seen that the only non-negligible terms in this sum are those for which $S$ and~$S'$ are supported on disjoint sets of variables (see
  e.g.\ the computation in~\S9.2 of~\cite{JansonLuczakRucinskiBk}).
  
  Whenever $S$ and $S'$ are supported on the disjoint sets of variables, a simple calculation shows that the expectation is $O( (\Exp X)^2 )$.
\end{proof}

This completes the proof of Theorem~\ref{thm:2-ub}(b).

\section{Conclusion and Outlook}\label{sec:conclusion}

What makes random $k$-rSAT intriguing is the presence of a second parameter next to $c = \nfrac{m}{n}$: the cardinality~$v$ of the truth-value set~$V$.  Since the
probability of satisfiability $p_k(n,m,v)$ increases with $v$ (Lemma~\ref{lem:p_monotone}), $c^{\pm}_k(v)$ increases with~$v$, too.  Based on computational
experiments, B\'ejar et al.~\cite{BejarManyaCabiscolFernandezGomes07} predicted that $c_3(v)$ increases logarithmically with~$v$.  This is clearly not the case, by
Theorem~\ref{thm:general-ub}.  However, based on B\'ejar et al.'s data, we conjecture that the functions $c^{\pm}_k(v)$ are strictly concave.

\begin{conjecture*}
  For all $k\ge 2$ and~$v\ge 2$, we have $c^{\pm}_k(v+1)-c^{\pm}_k(v) > c^{\pm}_k(v+2)-c^{\pm}_k(v)$.
\end{conjecture*}

Note that $p_k(n,cn,v)$ is not in general concave in~$v$.

B\'ejar et al.~\cite{BejarManyaCabiscolFernandezGomes07} conjecture (for $k=3$) that $c^+_k = c^-_k$, in other words, there is a threshold behaviour.  This can be
rephrased in terms of the dependence on the parameter~$v$:
\begin{question*}
  Does there exists a $v_k^*(c)$ such that $p_k(n,cn,v)=o(1)$ if $v < v^*_k(c)$, and $p_k(n,cn,v)=1-o(1)$ if $v > v^*_k(c)$?
\end{question*}
For example, in the case of 2-rSAT, we know that $p_2(n,\frac32 n, 2) = o(1)$ \cite{ChvatalReed92,DaLaVega92,Goerdt94}, and $p_2(n,\frac32n,\infty) = 1-o(1)$ by
Theorem~\ref{thm:2-lb}, but it is not clear whether the transition happens gradually, or suddenly, at some value $v^*$ between $2$ and~$\infty$.  If~$v^*$ exists,
though, then it must be ``finite'' in the sense that it does not depend on~$n$, cf.~Fig.~\ref{fig:transition}.

In the case of 2-rSAT, for $v=2$ there is a threshold at $c^-_2(2)=c^+_2(2)=1$, and for $v=\infty$, there is a threshold at $c^-_2(\infty)=c^+_2(\infty)=2$.  It
seems likely that this is true for the remaining values of~$v$, too.  In fact, we conjecture the following behaviour.
\begin{conjecture*}
  For all $\lambda=0,1,2,\dots,\infty$, we have
  \begin{equation*}
    c^-_2(2^\lambda+1)=c^+_2(2^\lambda+1) = 2 - 2^{-\lambda} = \sum_{j=0}^{\lambda} 2^{-j}.
  \end{equation*}
\end{conjecture*}

\providecommand{\bysame}{\leavevmode\hbox to3em{\hrulefill}\thinspace}
\providecommand{\MR}{\relax\ifhmode\unskip\space\fi MR }
\providecommand{\MRhref}[2]{%
  \href{http://www.ams.org/mathscinet-getitem?mr=#1}{#2}
}
\providecommand{\href}[2]{#2}

\end{document}